\documentclass[reprint,amsmath,amssymb,aps]{revtex4-1}

\usepackage{graphicx}
\usepackage{dcolumn}
\usepackage{bm}
\usepackage{amsthm}
\usepackage{braket}
\usepackage[colorlinks, linkcolor=blue]{hyperref}
\usepackage{qcircuit}
\usepackage{mathptmx}

\newtheorem{lemma}{Lemma}
\newtheorem{theorem}{Theorem}

\begin{document}

\preprint{APS/123-QED}

\title{Duality Quantum Computing with Subwave Projections}
\thanks{Department of Physics, Tsinghua University}%

\author{Fangjun Hu}
 \affiliation{State Key Laboratory of Low-Dimensional Quantum Physics and Department of Physics, Tsinghua University, Beijing 100084, China}
\author{Gui-Lu Long}%
 \email{gllong@tsinghua.edu.cn}
\affiliation{State Key Laboratory of Low-Dimensional Quantum Physics and Department of Physics, Tsinghua University, Beijing 100084, China}
\affiliation{Beijing National Research Center for Information Science and Technology, Beijing, 100084, China}
\affiliation{Collaborative Innovation Center of Quantum Matter, Beijing, 100084, China}
\affiliation{Beijing Academy of Quantum Information, Beijing, 100197, China}

\date{\today}

\pagestyle{empty}

\begin{abstract}
Duality quantum computing (DQC) offers the use of linear combination of unitaries (LCU), or generalized quantum gates, in designing quantum algorithms. DQC contains wave divider and wave combiner operations. The wave function of a quantum computer is split into several subwaves after the wave division operation. Then different unitary operations are performed on different subwaves in parallel. A quantum wave combiner combines the subwaves into a final wave function, so that a linear combination of the unitaries are performed on the final state. In this paper, we study of the properties of duality quantum computer with projections on subwaves. In subwave-projection DQC (SWP-DQC), we can realize the linear combinations of non-unitaries, and this not only gives further flexibility for designing quantum algorithms, but also offers additional speedup in the expected time complexity.  Specifically, SWP-DQC offers an $O(M)$ acceleration over  DQC with only final-wave-projection in the mean time complexity, where $M$ is the number of projections. As an application, we show that the ground state preparation algorithm recently proposed by Ge, Tura, and Cirac 
is actually an DQC algorithm, and we further  optimized the algorithm using SWP-DQC, which can save up to $\log_2 N$ qubits compared  DQC without subwave projection, where $N$ is the dimension of the system's Hilbert Space.
\end{abstract}

\maketitle


\section{Introduction}
\label{s1}

Duality quantum computing (DQC), which was put forward and developed in the last decade \cite{long512120general,gui2006general,gui2008duality,gui2009allowable}, offers the use of linear combination of unitaries (LCU) for quantum computing. Physically, DQC is a moving quantum computer passing through $M$-slits, which is called  the \textit{quantum wave divider} (QWD), where the wave function of the quantum computer is divided into different $M$ subwaves, and each subwave is processed differently.  The subwaves are recombined by a \textit{quantum wave combiner} (QWC) into a final wave function. The computing results are found by performing a measurement on the final state. The mathematical theory of DQC has been studied extensively  \cite{gudder2007mathematical,long2008mathematical,long2011duality,cao2010restricted,cao2013mathematical,zhang2010realization,chen2015generalized,chen2009note,zou2009mathematical,wu2014remarks}.
Gudder named linear combination of unitaries as {\it generalized quantum gates}, and he has shown that all bounded linear operators can be formed by the generalized quantum gates \cite{gudder2007mathematical}. Generalized quantum gate is also called LCU in Ref. \cite{childs2012hamiltonian}.

DQC is very useful in designing quantum algorithms. Previous quantum algorithms use only products of unitaries, for instance the quantum factorization algorithm \cite{shor1994algorithms} and the quantum search algorithm\cite{grover1996fast,long2001grover}. In DQC, the allowed quantum operations are extended to LCU, which offer more flexibility in constructing quantum algorithms.
It is interesting to note that though at first the coefficients in the generalized quantum gates (LCU) are positive numbers representing probabilities summing up to unity \cite{long512120general,gui2006general,gui2008duality,gudder2007mathematical}. It is soon found that in the most general form of DQC the coefficients in the LCU can have complex numbers  with a restriction that the sum of the modulus of the coefficients does not exceed unity \cite{gui2009allowable,cao2010restricted}.  Recently, people have witnessed a flood of works that use LCU with real and positive coefficients, to construct quantum algorithms
\cite{wan2007prime,harrow2009quantum,hao2010n,childs2012hamiltonian,berry2015simulating,wei2016duality1,wei2016duality2,wei2017realization,qiang2017quantum,qiang2018large,marshman2018passive,wei2018efficient}.
Notably, the HHL quantum algorithm for solving a set of linear equations \cite{harrow2009quantum} was shown to have used LCU \cite{wei2017realization}. DQC has also been used to construct secure remote quantum control by allowing different nodes performing different unitaries\cite{qiang2017quantum,qiang2018large}, and some passive quantum error correction code can be understood in terms of DQC \cite{marshman2018passive}. In these DQC algorithms, the coefficients are all real and positive. The application of DQC with general complex coefficients remains to be explored yet.

DQC also provides a realistic interpretation of quantum mechanics \cite{long2018realistic}, and in the description of processes of foundations of quantum mechanics, for instance,  the delayed-choice experiment\cite{roy2012nmr,xin2015realization,zhou2017quantum,long2018realistic,qin2018proposal,zhu2018beyond}, parity-time symmetric system \cite{zheng2013direct,zheng2018duality1,zheng2018duality2} and others \cite{huang2018pryce,liu2014realization,cui2012density}.

In Refs. \cite{long512120general,gui2006general} (toward the end of section 5), it is pointed out that the decomposition into subwaves can be iterated so that any subwave can be further decomposed into \textbf{sub-subwaves} to construct more complicated gates, such as linear combinations of non-unitaries.  Gudder showed that further divisions of sub-subwave \textit{cannot create new gate} \cite{gudder2007mathematical}.  In this work, we give Theorem \ref{t1} to explore the computability of DQC further in section \ref{s3}, which is a stronger conclusion and can imply Gudder's statement.

In this paper, we concentrate on another part of DQC, the projection of the wave function. In ordinary quantum computing, it is well-known that an measurement in the intermediate can always be postponed to the end of the calculation \cite{nielsen2002quantum}. However, it will make a difference in DQC because DQC is a probabilistic process. Usally,  only one projection measurement on final wave is performed in a DQC algorithm,  and a useful result can appear probabilistically.  We are able to apply projections on subwaves in the intermediate so as to give further flexibility and improvements in designing quantum algorithms. We find that further acceleration can be implemented by performing subwave-projections.

The paper is organized as follows. In section \ref{s2}, we briefly review DQC. In section \ref{s3}, we give a theorem that  proves the conclusion of Gudder on further divisions of subwaves in DQC. In section  \ref{s4}, we give the formalism of DQC with subwave-projections. In section \ref{s5}, we study the mathematical properties of SWP-DQC. In section \ref{s6}, we study some examples of SWP-DQC, namely, the ground state preparation algorithm proposed by Ge, Tura, and Cirac \cite{ge2017faster}. It is pointed out that the algorithm is a DQC algorithm, and we present an optimized version of the algorithm. A brief summary is given section \ref{s7}.

\section{Duality Quantum Computing}
\label{s2}

Duality quantum computing can create any linear combination of unitary operations. In this section, we give a brief introduction of DQC.  An $M$-slit  duality quantum computing can implement any linear combination of $M$ $(N\times N)$-unitary operations, where $N$ is the dimension of the Hilbert space of the quantum computer. $M$-slit can be implemented in an ordinary quantum computer with $m=\log_2 M$ qubits, or simply, a $d=M$  higher dimension qudit. Thus an $M$-slit duality quantum computer with $n$ qubits can be realized in an ordinary quantum computer with  $(m+n)$-qubit \cite{gui2008duality}.

In this article we only discuss \textit{sequential quantum circuit} realization of DQC. Namely, we only analyze the DQC circuits whose unitaries have to be operated one-by-one, instead of be manipulated parallelly. Here are the three main steps in DQC:

\textbf{QWD Step:} QWD is a unitary operation $V\otimes \mathbb{I}_n$, such that
\begin{equation}
(V\otimes \mathbb{I}_n)\ket{0}^{\otimes m} \ket{\Psi} = \sum_{i=0}^{2^m-1}  V_{i0} \ket{i} \ket{\Psi}.
\end{equation}
Here, the notation $\mathbb{I}_n$ represents the $2^n \times 2^n$ identity matrix on $n$ qubits.  The initial work qubit $\ket{\Psi}$, together with the auxiliary qubits $\ket{0}^{\otimes m}$, can be transformed to $\sum_{i=0}^{2^m-1}  V_{i0} \ket{i} \ket{\Psi}$ (which means that $V$ only acts on the auxiliary qubits), where $\sum_{i=0}^{M-1}|V_{i0}|^2=1$. The corresponding physics picture is that the spatial wave function is divided by $M$ slits, but the $n$ work qubits,  remain untouched.

\textbf{Parallel Operation Steps:} Then the unitary gates
\begin{eqnarray}
G_i &:=& \ket{i}\bra{i}\otimes U_i + \sum_{i'\neq i}^{M-1} \ket{i'}\bra{i'}\otimes \mathbb{I}_n\nonumber\\
&=& \ket{i}\bra{i}\otimes (U_i-\mathbb{I}_n)+\mathbb{I}_{m+n},
\end{eqnarray}
are performed on the subwaves.

$G_i$ is a $MN \times MN$ controlled-$U_i$ operations $(i=0,1,\dots,M-1)$, which maps a state with the form of $\ket{j}\ket{\psi}$, to
\begin{equation}
\resizebox{.9\hsize}{!}{$
G_i \ket{j}\ket{\psi}= \Bigg\{
\begin{array}{llc}
\ket{i}(U_i-\mathbb{I}_n)\ket{\psi}+\ket{i}\ket{\psi}=\ket{i}\bigl(U_i\ket{\psi}\bigl), & i=j, \\
\ket{j}\ket{\psi}, & i \neq j.
\end{array}$}
\end{equation}

The controlled gates in DQC are usually operated simultaneously, the quantum circuit here is sequential. After applying the set of controlled unitary operations $\left\{ G_0, G_1, \dots, G_{M-1} \right\}$ on the work qubits, namely the  subwaves, the following change is realized,
\begin{eqnarray}
\resizebox{.9\hsize}{!}{$G_{M-1} \dots G_1 G_0 \Bigg( \sum_{i=0}^{2^m-1}  V_{i0} \ket{i} \ket{\Psi} \Bigg) = \sum_{i=0}^{2^m-1} V_{i0} \ket{i} \big(U_i \ket{\Psi}\big)$}.
\end{eqnarray}

\textbf{QWC Step:}
QWC is also a unitary operation $W\otimes \mathbb{I}_n$ which only acts on auxiliary qubits, such that
\begin{equation}
(W\otimes \mathbb{I}_n) \ket{i} \ket{\psi} = \sum_{j=0}^{2^m-1} W_{ji} \ket{j} \ket{\psi}.
\end{equation}
Therefore, we have transformed the quantum system state to $\sum_{i=0}^{2^m-1} \sum_{j=0}^{2^m-1} V_{i0} W_{ji} \ket{j} \bigl(U_i \ket{\Psi}\bigr)$.
By measuring the $m$ auxiliary qubits, if the output is $\ket{0}^{\otimes m}$, then the work qubits state collapse to
$\sum_{i=0}^{2^m-1} V_{i0} W_{0i} \bigl(U_i \ket{\Psi}\bigr)$; otherwise, restart the algorithm. Before the projection on to $\ket{0}$ state of the auxiliary qubits, oblivious amplitude amplification can be performed  in order to increase the successful rate of the projection \cite{berry2015simulating}.

If the target work state is $\sum_{i=0}^{2^m-1} c_i U_i \ket{\Psi}$, the coefficient $c_i$ can be determined by choosing appropriate $W_{0i}$ and $V_{i0}$ (particularly, $W_{0i}=V_{i0}=\sqrt{c_i}$, if $c_i\ge 0$ and $\sum_i c_i=1$). An explicit construction of the QWD and QWC has been given by Zhang et al \cite{zhang2010realization}. It is worth noting that the expansion coefficients of LCU can be generally complex numbers \cite{gui2009allowable}. At present, duality quantum algorithms have used only LCU with real and positive coefficients.

The corresponding physics picture is that the $M$ subwaves are combined into $M$ subwaves. However, what we need is only the 0-th "channel" of the final $M$ channels on the right side. So we "\textbf{project out}" the 0-th channel by measuring the auxiliary qubits.

\section{Computability of DQC}
\label{s3}

To see the computability of DQC, we first give the following lemma \ref{l1} for preparation.


\begin{lemma}
For any $N \times N$ contracted Hermitian matrix $H$ ($\left\| H \right\| \le 1 $), the commutator $[H,(\mathbb{I}-H^2)^{1/2}]=0.$
\label{l1}
\end{lemma}
\begin{proof}
Rewrite $H$ into $Q^{\dag}AQ$, where $Q$ is a unitary matrix and $A=diag(\lambda_1,\lambda_2,...,\lambda_{N})$ whose $\lambda \in [-1,1]$ for contraction. Then, $(\mathbb{I}-H^2)^{1/2} = Q^{\dag}diag(\sqrt{1-\lambda^2_1},\dots,\sqrt{1-\lambda^2_{N}})Q.$ Immediately, examine that $H(\mathbb{I}-H^2)^{1/2}= (\mathbb{I}-H^2)^{1/2}H=Q^{\dag}~diag(\lambda_1\sqrt{1-\lambda^2_1}, \lambda_2\sqrt{1-\lambda^2_2}, \dots, \lambda_{N}\sqrt{1-\lambda^2_{N}})~Q$.
\end{proof}

Then the computability of DQC is ensured by our following Theorem \ref{t1}, which can be treated as an improvement of Ref. \cite{wu1994additive} (also, the Lemma 2.5 in Ref. \cite{wang2008note} gave a weaker version of Theorem \ref{t1}):

\begin{theorem}
For any $N \times N$ contraction matrix $A$ ($\left\| A \right\| \le 1 $), $A$ can always be decomposed into only two unitary operations averagely, i.e. $A = \frac{1}{2}U_0 + \frac{1}{2}U_1$.
\label{t1}
\end{theorem}
\begin{proof}
Consider the polar decomposition of $A=UP$, where $U$ is a unitary matrix and $P$ is a positive-semidefinite Hermitian, such that $P=(A^\dag A)^{1/2}$. $P$ is also a contraction. Define
\begin{equation}
U_0 = U \bigl(P+i(\mathbb{I}-P^2)^{1/2}\bigr),~~~~U_1 = U \bigl(P-i(\mathbb{I}-P^2)^{1/2}\bigr).
\end{equation}
Notice that
\begin{eqnarray}
\big[P+i(\mathbb{I}-P^2)^{1/2}\big]\big[P-i(\mathbb{I}-P^2)^{1/2}\big] &=& \mathbb{I}, \\
\big[P-i(\mathbb{I}-P^2)^{1/2}\big]\big[P+i(\mathbb{I}-P^2)^{1/2}\big] &=& \mathbb{I},
\end{eqnarray}
because $[P,(\mathbb{I}-P^2)^{1/2}]=0$ by Lemma \ref{l1}. On the other hand, $P \pm i(\mathbb{I}-P^2)^{1/2}$ are unitaries, since $\big[P-i(\mathbb{I}-P^2)^{1/2}\big]^\dag=\big[P+i(\mathbb{I}-P^2)^{1/2}\big]$. Obviously, $A = \frac{1}{2}U_0 + \frac{1}{2}U_1$.
\end{proof}

It is a stronger conclusion. As what we have shown, every finite-dimension contraction matrix can always be decomposed into \textbf{only two unitary operations averagely}, as given in Ref. \cite{gui2008duality}. In practice, whether we can construct a linear decomposition of a given non-unitary gate in terms of linear combination of unitaries is not a problem, what we really concern is how to create a linear decomposition of a given non-unitary gate using certain gates group, e.g. the direct product of Pauli matrices.

\section{Duality Quantum Computation with Subwave-Projection }			
\label{s4}

It was already pointed out that DQC can also realize linear combinations of non-unitaries(LCNU) by using wave divider in the subwaves \cite{gui2006general}. Using LCNU cannot create new gates \cite{gudder2007mathematical}, but it provides more flexibility in designing quantum algorithms.  The result of the DQC is obtained by first  projecting out the auxiliary qubits into $\ket{0}$ states. Here we study the case where we perform projection of the subwaves, namely DQC with subwave projections, SWP-DQC. It gives another way to construct linear combination of non-unitary operations. Considering  the initial $n$-qubit quantum state $\ket{\Psi}$. We divide the auxiliary qubits into two groups in SWP-DQC: the \textit{first group of $m$ auxiliary qubits} initially in $\ket{0}^{\otimes m}$ and the \textit{second group of $p$ auxiliary qubits} initially in $\ket{0}^{\otimes p}$.

\begin{figure*}[htb]
\centering
\mbox{
\Qcircuit @C=1.em @R=0.7em {
& \lstick{\ket{0}} &\qw & \multigate{3}{V} & \ctrlo{0} \qw & \qw & \ctrlo{0} \qw &\qw & \dots & & \ctrl{0} \qw & \multigate{3}{W} & \measureD{\ket{0}}\\
& \lstick{\vdots~} &\qw{\backslash}^{m-3} & \ghost{V} & \qw{\vdots} & \qw & \qw{\vdots} & \qw & {\dots} & & \qw{\vdots} & \ghost{W} & \qw \vdots \\
& \lstick{\ket{0}} &\qw& \ghost{V} & \ctrlo{1} \qw & \qw & \ctrlo{1} \qw &\qw & \dots & & \ctrl{1} \qw & \ghost{W} & \measureD{\ket{0}}\\
& \lstick{\ket{0}} & \qw & \ghost{V}  & \ctrlo{1} & \qw \qw & \ctrl{1} \qw &\qw & \dots & & \ctrl{1} \qw & \ghost{W} & \measureD{\ket{0}}
	\inputgroupv{1}{4}{2.0em}{2.5em}{m~qubits~~~~~~~~~~~~~~~~~~~~~} \\
& \lstick{p~qubits~\ket{0}^{\otimes p}~~~} &\qw{\backslash}^{p} & \qw & \multigate{4}{U_0} & \measureD{\ket{0}^{\otimes p}} & \multigate{4}{U_1} & \measureD{\ket{0}^{\otimes p}} & \dots & & \multigate{4}{U_{M-1}} & \qw & \measureD{\ket{0}^{\otimes p}}\\
&  &\qw & \qw & \ghost{U_0} & \qw & \ghost{U_1}  &\qw & \dots & &\ghost{U_{M-1}} & \qw\\
&  &\qw{\backslash}^{n-3} & \qw & \ghost{U_0} & \qw & \ghost{U_1}  &\qw & \dots & &\ghost{U_{M-1}} & \qw & ~~~\sum_{i=0}^{M-1} c_i B_i \ket{\Psi} \\
&  &\qw & \qw & \ghost{U_0} & \qw & \ghost{U_1} &\qw & \dots & &\ghost{U_{M-1}} & \qw & ~~~ = A \ket{\Psi}\\
&  &\qw & \qw & \ghost{U_0} & \qw & \ghost{U_1} &\qw & \dots & &\ghost{U_{M-1}} & \qw
	\inputgroupv{6}{9}{2.0em}{2.5em}{n~qubits~\ket{\Psi}~~~~~~~~~~~~~~~~~} }
}
\caption{Illustration of quantum circuit of SWP-DQC. The first group of $m$ auxiliary qubits and the second group  $p$ auxiliary qubits are in the $\ket{0}^m$ and $\ket{0}^p$ states initially.   $\ket{\chi}$ in each D-shaped measurement box means obtaining qubit $\ket{\chi}$ after measuring. In SWP-DQC, all subwave-projections are onto the $\ket{0}^p$ state. On the righthand, one obtains $\sum_i c_i B_j$, where $B_j$ is a non-unitary operator.}\label{f1}
\end{figure*}
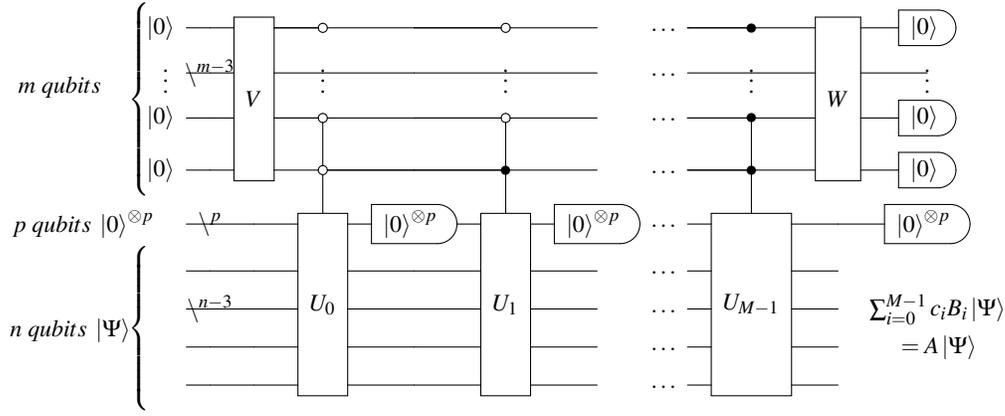

In SWP-DQC,  quantum state $\ket{0}^{\otimes (m+p)}\ket{\Psi}$ is transformed to $\sum_{i=0}^{2^m-1}  V_{i0} \ket{i} \ket{0}^{\otimes p} \ket{\Psi}$ by the QWD operation $V\otimes \mathbb{I}_{p+n}$. Instead of constructing a linear combination of $M$ unitary operations $U_0,U_1,...,U_{M-1}$ on the work qubits of $M$ subwaves $V_{i0}\ket{\Psi}$, we wish to construct a linear combination of several  non-unitary operations $B_0,B_1,...,B_{M-1}$ on the work qubits,
\begin{equation}
\label{e9}
A \ket{\Psi} = \sum_{i=0}^{M-1} c_i B_i \ket{\Psi}.
\end{equation}
where $A$ is a contraction. Precisely, consider an $MPN \times MPN$ controlled-$U_i$ operation $G_i := \ket{i}\bra{i}\otimes (U_i-\mathbb{I}_{p+n})+\mathbb{I}_{m+p+n}$, where $U_i$ is a $PN \times PN$ unitary  matrix, which can implement the effect of non-unitary operation $B_i$ on arbitrary state $\ket{0}^{\otimes p} \ket{\psi}$ by:
\begin{equation}
U_i \ket{0}^{\otimes p} \ket{\psi} = \ket{0}^{\otimes p}\bigl(B_i \ket{\psi}\bigr)+\sum_{k=1}^{2^p-1} \ket{k}\ket{\psi_{ik}}.
\end{equation}

In SWP-DQC algorithm, we only need to retain the partial terms $\ket{0}^{\otimes p}\bigl(B_i \ket{\Psi}\bigr)$ of the operated subwave by projection. More precisely, here we describe the algorithm in details.

\textbf{Basic Step:} As the $0$-th step, performing the controlled-$U_0$ unitary operation $G_0$ on $\sum_{i=0}^{2^m-1} V_{i0} \ket{i} \ket{0}^{\otimes p} \ket{\Psi}$, yields the state
\begin{equation}		
\resizebox{1.\hsize}{!}{$V_{00}\ket{0}^{\otimes m+p}B_0 \ket{\Psi}+V_{00}\ket{0}^{\otimes m}\sum_{k=1}^{2^p-1} \ket{k}\ket{\Psi_{0k}}+\sum_{i=1}^{2^m-1} V_{i0} \ket{i} \ket{0}^{\otimes p} \ket{\Psi}$}.
\end{equation}
Measure the second auxiliary qubits, the probability of reading out $\ket{0}^{\otimes p}$ is
\begin{eqnarray}
\label{e12}
p_0 &=& \frac{b_0 |V_{00}|^2 + |V_{10}|^2 + \dots + |V_{M-1,0}|^2}{|V_{00}|^2 + |V_{10}|^2 + \dots + |V_{M-1,0}|^2} \nonumber \\
&=& b_0 |V_{00}|^2 + |V_{10}|^2 + \dots + |V_{M-1,0}|^2,
\end{eqnarray}
where we difine $p_i$ the success probability of the $i$-th step, and we also define the coefficient
\begin{equation}
b_i=\bra{\Psi}B_i^\dag B_i\ket{\Psi}.
\end{equation}
If the output is $\ket{0}^{\otimes p}$, then we obtain the intermediate quantum state
\begin{equation}		
V_{00}\ket{0}^{\otimes m} \ket{0}^{\otimes p}\bigl(B_0 \ket{\Psi}\bigr)+\sum_{i=l}^{2^m-1}  V_{i0} \ket{i} \ket{0}^{\otimes p} \ket{\Psi},
\end{equation}
and continue.  Otherwise, restart the algorithm again. Assume the 0-th step costs a run time of $t_0$.

\textbf{Induction Steps:} If the previous $l$-th steps are successful, which means that we have obtained the quantum state
\begin{equation}
\sum_{i=0}^{l-1} V_{i0}\ket{i} \ket{0}^{\otimes p}\bigl(B_i \ket{\Psi}\bigr)+\sum_{i=l}^{2^q-1}  V_{i0} \ket{i} \ket{0}^{\otimes p} \ket{\Psi}.
\end{equation}
Then perform the controlled-$U_l$ unitary operation $G_l$ on the system, yielding
\begin{eqnarray}		
& & \sum_{i=0}^{l-1} V_{i0}\ket{i} \ket{0}^{\otimes p}\bigl(B_i \ket{\Psi}\bigr)+V_{l0}\ket{l} \ket{0}^{\otimes p}\bigl(B_l \ket{\Psi}\bigr) \nonumber \\
&+&  V_{l0}\ket{l}\sum_{k=1}^{2^p-1} \ket{k}\ket{\Psi_{lk}}+\sum_{i=l+1}^{2^m-1}  V_{i0} \ket{i} \ket{0}^{\otimes p} \ket{\Psi}.
\end{eqnarray}
Measure the second auxiliary qubits, the probability of obtaining $\ket{0}^{\otimes p}$ is
\begin{equation}
\label{e17}
p_l = \frac{b_0 |V_{00}|^2 + \dots + b_{l-1} |V_{l-1,0}|^2 + b_{l} |V_{l,0}|^2 + \dots + |V_{M-1,0}|^2}{b_0 |V_{00}|^2 + \dots + b_{l-1} |V_{l-1,0}|^2 + |V_{l,0}|^2 + \dots + |V_{M-1,0}|^2}.
\end{equation}

If the output is $\ket{0}^{\otimes p}$, then we obtain an intermediate quantum state
\begin{equation}
\sum_{i=0}^{l} V_{i0}\ket{i} \ket{0}^{\otimes p}\bigl(B_i \ket{\Psi}\bigr)+\sum_{i=l+1}^{2^m-1}  V_{i0} \ket{i} \ket{0}^{\otimes p} \ket{\Psi}
\end{equation}
and continue. Otherwise, restart the algorithm. Assume the $l$-th step uses a run time of $t_l$.

\textbf{Final Step:} If all the $M=2^m$ steps succeed, we have done $M$ times projection measurement and obtained $\sum_{i=0}^{M-1} V_{i0}\ket{i} \ket{0}^{\otimes p}\bigl(B_i \ket{\Psi}\bigr)$. Perform the QWC operation on these $M$ subwaves, we obtain
\begin{eqnarray}		%
& & \sum_{i=0}^{2^m-1}\sum_{j=0}^{2^m-1} W_{ji}V_{i0}\ket{j}\ket{0}^{\otimes p}\bigl(B_i \ket{\Psi}\bigr) \nonumber\\
&=& \ket{0}^{\otimes (m+p)} \sum_{i=0}^{2^m-1} W_{0i}V_{i0} B_i \ket{\Psi} \nonumber\\
&+& \sum_{i=0}^{2^m-1}\sum_{j=1}^{2^m-1} W_{ji}V_{i0}\ket{j}\ket{0}^{\otimes p}\bigl(B_i \ket{\Psi}\bigr).
\end{eqnarray}
Measure the first auxiliary qubits, if the output is $\ket{0}^{\otimes m}$, then the $n$-qubit quantum state has been transformed to the target state $\sum_{i=0}^{M-1} c_i B_i \ket{\Psi} = A \ket{\Psi}$. The coefficient $c_i$ can be determined by choosing appropriate $W_{0i}$ and $V_{i0}$. In this work, we restrict ourselves to real and positive $c_i \ge 0$, and this implies $W_{0i}=V_{i0}=\sqrt{c_i}$. The success probability of the final step is
\begin{equation}
\label{e20}
p_M = \frac{\bra{\Psi} A^\dag A \ket{\Psi}}{b_0 |V_{00}|^2 + \dots + b_{M-1} |V_{M-1,0}|^2}.
\end{equation}
The quantum circuit of SWP-DQC is illustrated in Figure \ref{f1}. The corresponding conceptual physics picture of a SWP-DQC device is shown in Figure \ref{f2}.

\begin{figure}[htb]
\begin{center}
\includegraphics[width=0.5\textwidth]{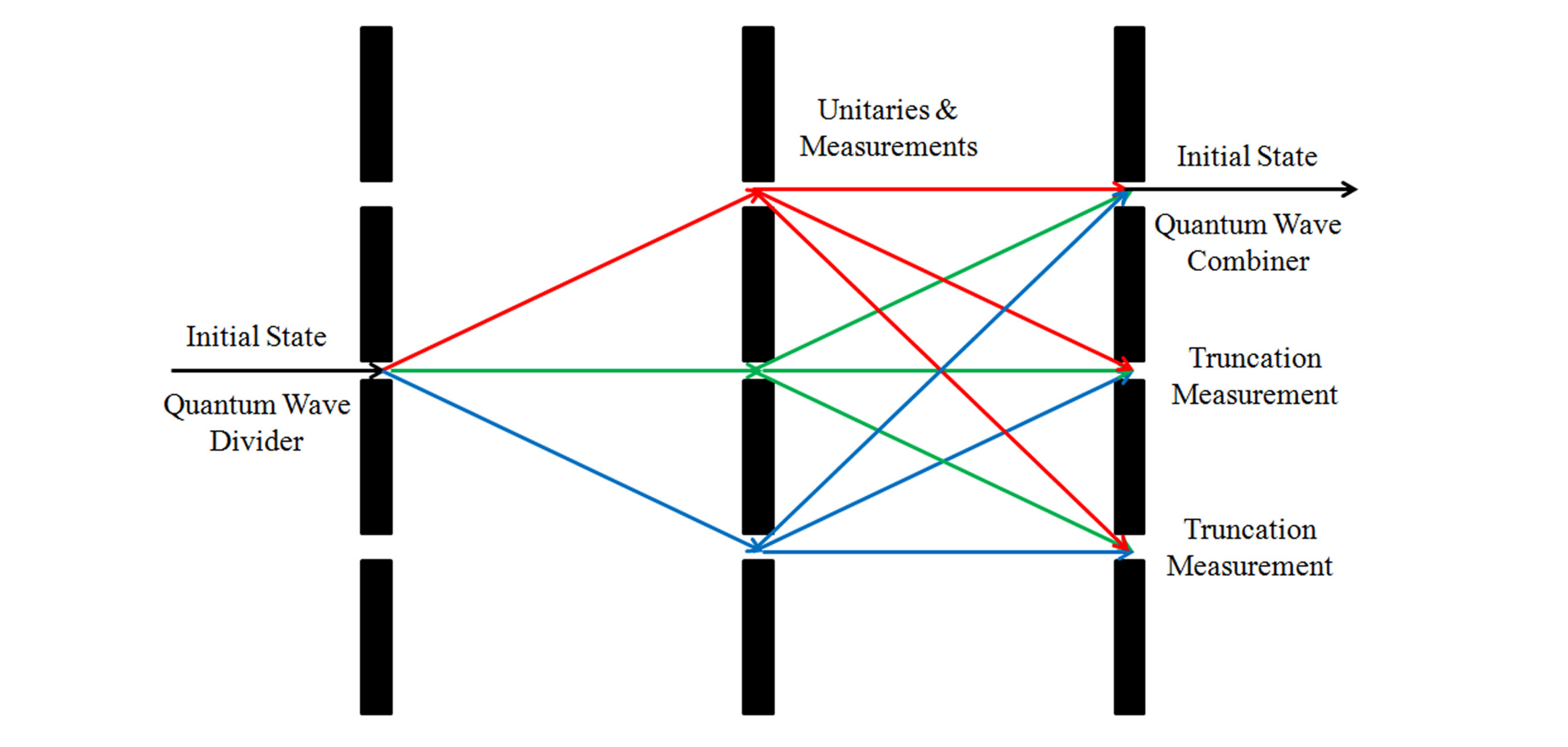}
\caption{The conceptual physical picture of a 3-slits SWP-DQC. In the middle, the subwaves are manipulated and some auxiliary qubits are projected on to $\ket{0}$ states. }\label{f2}
\end{center}
\end{figure}

\section{Some Mathematical Results of SWP-DQC}
\label{s5}

Here we focus on two mathematical results of SWP-DQC: the success probability and the mean time complexity. We show that SWP-DQC can give a polynomial acceleration compared to DQC with only a final projection.

\subsection{Success Probability}
Our first mathematical result can be derived by calculating the success probability of the entire algorithm from Eq. (\ref{e12}), Eq. (\ref{e17}) and Eq. (\ref{e20}),
\begin{eqnarray}
P &=& p_0 \times p_1 \times \dots \times p_{M-1} \times p_M \nonumber \\
&=& \big( b_0 |V_{00}|^2 + |V_{10}|^2 + \dots + |V_{M-1,0}|^2 \big) \nonumber \\
&\times& \frac{b_0 |V_{00}|^2 + b_{1} |V_{10}|^2 + \dots + |V_{M-1,0}|^2}{v_0 |V_{00}|^2 + |V_{10}|^2 + \dots + |V_{M-1,0}|^2} \times \dots \nonumber \\
&\times& \frac{b_0 |V_{00}|^2 + \dots +  b_{M-2} |V_{M-2,0}|^2 + b_{M-1} |V_{M-1,0}|^2}{b_0 |V_{00}|^2 + \dots + b_{M-2} |V_{M-2,0}|^2 + |V_{M-1,0}|^2} \nonumber \\
&\times& \frac{\bra{\Psi} A^\dag A \ket{\Psi}}{b_0 |V_{00}|^2 + \dots + b_{M-1} |V_{M-1,0}|^2} \nonumber\\
&=& \bra{\Psi} A^\dag A \ket{\Psi}.
\end{eqnarray}
The expression of the probability is natural because it is just the norm of the vector $A\ket{\Psi}$, and is less than $1$ because we have assumed that $A$ is a contraction. In DQC with final-wave-projection, the target state is also $\sum_{i=0}^{M-1} c_i B_i \ket{\Psi} = A \ket{\Psi}$. Therefore, as a probabilistic algorithm, its success probability must also be the square of modulus,
\begin{equation}
P'= \bra{\Psi} A^\dag A \ket{\Psi}.
\end{equation}

\subsection{Mean Time Complexity and SWP-DQC Acceleration}
In order to analyze the mean time complexity, we make the following analysis on a model in probability theory. Suppose an event occurs with probability of $p$ in an experiment. It will be terminated once the experiment is successful. If it fails, then we make another experiment. We stop until we succeed to get the event. The probability of success after $s$-time is defined as $P(s) = p(1-p)^{s-1}$. The mean number of experiments will be
\begin{equation}
E = \sum_{s=1}^\infty s \cdot p(1-p)^{s-1} = \lim_{i \to \infty} \frac{1-(1-p)^s(1+s \cdot p)}{p} = \frac{1}{p}.
\end{equation}
The expected value of run time of SWP-DQC can be expressed by the known physical quantity from this model in probability theory. In SWP-DQC, there are several projections on the subwaves, each projection succeeds with a probability $p_i$. Because the call of the $i$-th step in the algorithm is equivalent to an event with success probability of $p_i p_{i+1}\dots p_M$ in a series of experiments. Therefore, the overall mean time of SWP-DQC is the sum of mean time of all steps:
\begin{eqnarray}
\label{e24}
Et &=& \frac{t_0}{p_0 p_1 \dots p_M} + \frac{t_1}{p_1 \dots p_M} + \dots + \frac{t_{M-1}}{p_{M-1} p_M} \\ \nonumber
&=& \frac{t_0 + p_0t_1 + (p_0p_1)t_2 + \dots + (p_0p_1\dots p_{M-2})t_{M-1}}{p_0  p_1 \dots p_{M-1} p_M}.
\end{eqnarray}
Notice that the expectation depends on the order of gates. The optimal order requires complicated \textit{numerical calculation}.

However, the total time of DQC with final-wave-projection is obviously $t_0 + t_1 + t_2 + \dots + t_{M-1}$, with the success probability of $P'= \bra{\Psi} A^\dag A \ket{\Psi}$. Thus, the expected value of run time in DQC with final-wave-projection can be derived of our model in probability theory:
\begin{eqnarray}
Et' &=& \frac{t_0 + t_1 + t_2 + \dots + t_{M-1}}{\bra{\Psi} A^\dag A \ket{\Psi}} = \frac{t_0 + t_1 + t_2 + \dots + t_{M-1}}{p_0  p_1 \dots p_{M-1} p_M}\\
&>& \frac{t_0 + p_0t_1 + (p_0p_1)t_2 + \dots + (p_0p_1\dots p_{M-2})t_{M-1}}{p_0  p_1 \dots p_{M-1} p_M} =  Et. \nonumber
\end{eqnarray}

Even though the numerical calculation in Eq. (\ref{e24}) is very troublesome, we can still compare the complexity between $Et$ and $Et'$ with some simplified assumptions. Assume that each basic gate has the same time complexity $t_0=t_1= \dots =t_{M-1}=1$, and the same order of success probability of $p$ (recall the Eq. (\ref{e17}), in each step $p_i$ is very close to 1 if $M$ is large enough). Then we can derive that $Et' = O(\frac{M}{p^M})$, whereas
\begin{equation}
Et= \frac{1+p+\dots+p^{M-1}}{p^{M+1}} = \frac{1-p^M}{p^M p(1-p)} = O\left(\frac{1}{p^M}\right).
\end{equation}
Our analysis on the acceleration is only valid in  \textit{sequential} realization of SWP-DQC. It is interesting to study the parallel realization of SWP-DQC, and study its acceleration.

Another point is that, in most cases, $M$ increases rapidly as precision of calculation gets higher. For example, in order to get a higher precision, the larger the evolution time $t$ in a quantum algorithm, the more steps we need to use, which means that $p$ may still be the same, but the complexity of the algorithm increases with the matrix number $M$. In this respect, we say,  SWP-DQC has an $O(M)$ speedup compared with DQC with final-wave-projection in time complexity.

\section{Application: An optimization of ground state preparation quantum algorithm}			
\label{s6}

Yimin Ge, Jordi Tura, and J. Ignacio Cirac recently proposed a general-purpose quantum algorithm for preparing ground states of a quantum Hamiltonian from a given trial state (we will use GTC algorithm hereafter). Here we show that the GTC algorithm is a DQC algorithm, and we also give an optimization of GTC algorithm by using the SWP-DQC. The optimized algorithm uses $2+\log_2 N$ less qubits, where $N$ is the dimension of the Hermitian matrix.

\subsection{Brief Description of GTC Algorithm}
\label{ss1}

Here is a brief description of GTC algorithm for ground state preparation \cite{ge2017faster}. For the $N \times N$ Hermitian matrix $\tilde H$, assume $n=\log_2 N$, and the spectrum of $\tilde H$ lies in $[0,1]$, with the lowest eigenvalue $\lambda_0$. The spectrum is assumed to be  non-degenerate. Let $E \in [0,\lambda_0]$ be a known real number, then define $\delta_E:=\lambda_0 -E$ and $H:=(1+E)~\mathbb{I}_n-\tilde H$. Namely, $H$'s spectrum lies in $[E,1-\delta_E]$. And also assume that all other eigenvalues of $H$ are $\le 1- \delta_E -\Delta$. The core ideal of this algorithm is that the iteration of $H$ is almost a projector onto the ground state, because the power of $1-\delta_E$ is far larger than the power of the other eigenvalues. Ge et al showed that by iteratively performing $H$ on the trial state $\ket{\phi}=\phi_0 \ket{\lambda_0} + \ket{\lambda^\bot_0}$ for $M_0$ times, where
\begin{eqnarray}
M_0=O(\frac{1}{\Delta}\log \frac{1}{\left| \phi_0 \right| \epsilon}),
\end{eqnarray}
which is an even integer, the norm of the difference between normalized state $H^{M_0}\ket{\phi}$ and $\ket{\lambda_0}$ will be less than $\epsilon$.

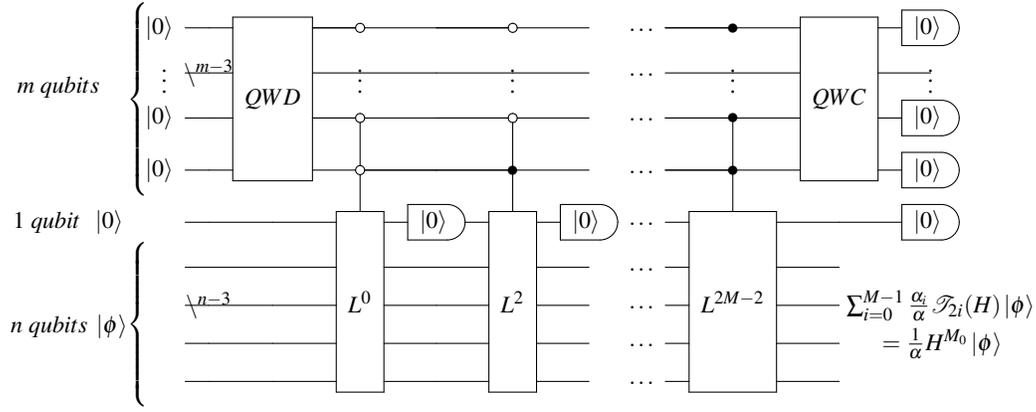
\begin{figure*}[htb]
\centering
\mbox{
\Qcircuit @C=1.em @R=0.7em {
& \lstick{\ket{0}} &\qw & \multigate{3}{QWD} & \ctrlo{0} \qw & \qw & \ctrlo{0} \qw &\qw & \dots & & \ctrl{0} \qw & \multigate{3}{QWC} & \measureD{\ket{0}}\\
& \lstick{\vdots~} &\qw{\backslash}^{m-3} & \ghost{QWD} & \qw{\vdots} & \qw & \qw{\vdots} & \qw & {\dots} & & \qw{\vdots} & \ghost{QWC} & \qw \vdots \\
& \lstick{\ket{0}} &\qw& \ghost{QWD} & \ctrlo{1} \qw & \qw & \ctrlo{1} \qw &\qw & \dots & & \ctrl{1} \qw & \ghost{QWC} & \measureD{\ket{0}}\\
& \lstick{\ket{0}} & \qw & \ghost{QWD}  & \ctrlo{1} & \qw \qw & \ctrl{1} \qw &\qw & \dots & & \ctrl{1} \qw & \ghost{QWC} & \measureD{\ket{0}}
	\inputgroupv{1}{4}{2.0em}{2.5em}{m~qubits~~~~~~~~~~~~~~~~~~~~~} \\
& \lstick{1~qubit~~\ket{0}~~~~~~~~} &\qw & \qw & \multigate{4}{L^0} & \measureD{\ket{0}} & \multigate{4}{L^2} & \measureD{\ket{0}} & \dots & & \multigate{4}{L^{2M-2}} & \qw & \measureD{\ket{0}}\\
&  &\qw & \qw & \ghost{L^0} & \qw & \ghost{L^2}  &\qw & \dots & &\ghost{L^{2M-2}} & \qw\\
&  &\qw{\backslash}^{n-3} & \qw & \ghost{L^0} & \qw & \ghost{L^2}  &\qw & \dots & &\ghost{L^{2M-2}} & \qw & ~~~\sum_{i=0}^{M-1} \frac{\alpha_i}{\alpha} \mathcal{T}_{2i}(H) \ket{\phi}\\
&  &\qw & \qw & \ghost{L^0} & \qw & \ghost{L^2} &\qw & \dots & &\ghost{L^{2M-2}} & \qw & ~~~ = \frac{1}{\alpha} H^{M_0}\ket{\phi}\\
&  &\qw & \qw & \ghost{L^0} & \qw & \ghost{L^2} &\qw & \dots & &\ghost{L^{2M-2}} & \qw
	\inputgroupv{6}{9}{2.0em}{2.5em}{n~qubits~\ket{\phi}~~~~~~~~~~~~~~~~~} }
}
\caption{SWP-DQC for optimized GTC Algorithm, where the first column and row of QWD and QWC is $V_{i0} = W_{0i} = \sqrt{\frac{\alpha_i}{\alpha}}$ and $L$ is defined in Eq. (\ref{t2}). We only need one qubit in the second group of auxiliary qubit.}
\label{f3}
\end{figure*}

The $M_0$ power of $H$ can be calculated using a linear combination of $M$ terms of non-unitary operations to a good approximation, where
\begin{equation}
\label{e27}
M=O(\sqrt{\frac{1}{\Delta}}\log \frac{1}{\left| \phi_0 \right| \epsilon}),
\end{equation}
namely,
\begin{eqnarray}
H^{M_0}=\sum_{i=0}^{M-1} \alpha_{i} \mathcal{T}_{2i}(H)+O(\left| \phi_0 \right| \epsilon),
\end{eqnarray}
where $\mathcal{T}_{2i} (x)$ represents the $2i$-th Chebyshev polynomials of the first kind, $\alpha_i$ is defined as $2^{1-2m_0} \bigl( \frac{1}{2} \bigr)^{\delta_{i0}} \binom{2m_0}{m_0+i}$, $m_0$ is $M_0/2$, and $\delta_{i0}$ is the Kronecker delta.

The operation $\sum_{i=0}^{M-1} \alpha_{i} \mathcal{T}_{2i}(H)$ is the linear combination of $M$ non-unitary operations $\mathcal{T}_{2i}(H)$, where $M$ is defined in Eq. (\ref{e27}). It can be obtained by using the following matrix $L$ and its powers
\begin{equation}
L = \left(
\begin{array}{ccc}
H & -(\mathbb{I}-H^2)^{1/2} \\
(\mathbb{I}-H^2)^{1/2} & H
\end{array}
\right),\label{t2}
\end{equation}
and
\begin{equation}
L^n = \left(
\begin{array}{ccc}
\mathcal{T}_n(H) & -(\mathbb{I}-H^2)^{1/2}\mathcal{U}_{n-1}(H) \\
(\mathbb{I}-H^2)^{1/2}\mathcal{U}_{n-1}(H) & \mathcal{T}_n(H)
\end{array}
\right).\label{t2p}
\end{equation}

\subsection {Quantum Circuit for Optimization of GTC Algorithm}
\label{ss3}

The matrix $L$ appearing in Eq. (\ref{t2}) was applied by Ge et al to construct a quantum walk in a larger Hilbert space, namely they doubled the entire system and treated $L$ as an operator on the Hilbert space of $\mathbb{C}^{2N}\otimes\mathbb{C}^{2N}$, which has a dimension of $4N^2$. This costs too much qubit resource to complete the task, which adds in $1+\log_2 (4N^2/N) = 3 + \log_2 N$ more auxiliary qubits. The GTC ground state preparation algorithm developed requires $2\log_2 N +\log_2 M+3$ qubits.

Here we give an optimal algorithm which only uses $\log_2 N +\log_2 M+1$ qubits, that is,  $2+\log_2 N$ less qubits, by using SWP-DQC. In our optimized algorithm, it is applied in the Hilbert space $\mathbb{C}^{2}\otimes\mathbb{C}^{2N}$, instead of $\mathbb{C}^{2N}\otimes\mathbb{C}^{2N}$.

Besides the $m=\log_2 M$ auxiliary qubits in the \textit{first group of auxiliary qubits} and the $n=\log_2 N$ work qubits, we add another auxiliary qubit ($p=1$, in the \textit{second group of auxiliary qubit}). The QWD and QWC  can be constructed explicitly as,
\begin{equation}
V_{i0} = W_{0i} = \sqrt{\frac{\alpha_i}{\alpha}},
\end{equation}
where $\alpha=\sum_{i=0}^{M-1} \alpha_i$ (here we have already had $\alpha_i\ge 0$), such that the QWD map $\ket{0}^{\otimes m} \ket{0} \ket{\phi}$ to $\sum_{i=0}^{M-1} \sqrt{\frac{\alpha_i}{\alpha}} \ket{i} \ket{0} \ket{\phi}$.

For the second group of auxiliary qubit and the work qubits, the quantum state $\ket{0} \ket{\phi}$ in the Hilbert space of $\mathbb{C}^2\otimes\mathbb{C}^{N}$ can be written in a matrix form,
\begin{equation}
\ket{0} \otimes \ket{\phi} = \left(
\begin{array}{c}
1 \\
0
\end{array}
\right) \otimes \boldsymbol{\phi} = \left(
\begin{array}{c}
\boldsymbol{\phi} \\
0
\end{array}
\right).
\end{equation}
Notice that $\phi$ is a state, whereas $\boldsymbol{\phi}$ is a complex vector.

Then the $2N \times 2N$ unitary matrix $L^{2i}$ in Eq. (\ref{t2}) maps $\ket{0} \otimes \ket{\phi}$ to
\begin{eqnarray}	
& & \left(
\begin{array}{cc}
\mathcal{T}_n(H) & -(\mathbb{I}-H^2)^{1/2}\mathcal{U}_{n-1}(H) \\
(\mathbb{I}-H^2)^{1/2}\mathcal{U}_{n-1}(H) & \mathcal{T}_n(H)
\end{array}
\right) \left(
\begin{array}{c}
\boldsymbol{\phi} \\
0
\end{array}
\right) \nonumber \\
&=& \left(
\begin{array}{c}
\mathcal{T}_n(H) \boldsymbol{\phi} \\
(\mathbb{I}-H^2)^{1/2}\mathcal{U}_{n-1}(H) \boldsymbol{\phi}
\end{array}
\right).
\end{eqnarray}


The trial state of work qubits, together with the state of the $M+1$ auxiliary qubits, have been transformed to $M$ subwaves $\sum_{i=0}^{M-1} \sqrt{\frac{\alpha_i}{\alpha}} \ket{i} \ket{0} \ket{\phi}$ by QWD. Repeatedly apply the  $2MN \times 2MN$ controlled-$L^{2i}$ operation $G_i := \ket{i}\bra{i}\otimes (L^{2i}-\mathbb{I}_{1+n})+\mathbb{I}_{m+1+n}$ ($i=0,1,...,M-1$) and the projection measurements $\mathbb{I}_m \otimes \ket{0}\bra{0} \otimes \mathbb{I}_n$ for $M$ times, then the quantum system state becomes
\begin{equation}
\sum_{i=0}^{M-1} \sqrt{\frac{\alpha_i}{\alpha}} \ket{i} \ket{0} \mathcal{T}_{2i} (H)\ket{\phi}.
\end{equation}
Finally, perform the QWC, we get the state
\begin{equation}
\sum_{i=0}^{M-1}\sum_{j=0}^{{M-1}} W_{ji} \sqrt{\frac{\alpha_i}{\alpha}} \ket{j} \ket{0} \mathcal{T}_{2i} (H)\ket{\phi}.
\end{equation}
Use the projection measurement $\ket{0^m}\bra{0^m} \otimes \mathbb{I} \otimes \mathbb{I}_n$ again. If the output is $\ket{0}^{\otimes p}$, then the final state of the entire quantum system will collapse to $\sum_{i=0}^{M-1}W_{0i} \sqrt{\frac{\alpha_i}{\alpha}} \ket{0}^{\otimes m} \ket{0} \mathcal{T}_{2i} (H)\ket{\phi} = \sum_{i=0}^{M-1} \frac{\alpha_i}{\alpha} \ket{0}^{\otimes m} \ket{0} \mathcal{T}_{2i} (H)\ket{\phi}$. The quantum circuit  of this SWP-DQC optimized algorithm is shown in Figure \ref{f3}.

Note that as the precision becomes higher, GTC algorithm requires more number of Chebyshev polynomials $M$, and the mean time required by GTC algorithm becomes larger. The mean time required by our optimized algorithm requires only $O(1/M)$ of that of the GTC algorithm, as discussed in \ref{s4}.

\section{Summary}			
\label{s7}

In this article, we presented a DQC with subwave projections, the SWP-DQC. Explicit quantum circuit of SWP-DQC is constructed.  We proved that the mean time complexity has an $O(M)$ acceleration compared to DQC with only final-wave-projection. We also find the run time depends on the orders of the controlled gates, and this is especially important in the future in constructing  programs for concrete problems.

As an application, we show that the ground state preparation proposed by Ge, Tura, and Cirac is a DQC algorithm. We constructed an optimization of GTC algorithm, and  it not only saves $(2+\log_2 N)$ qubits, but also provides additional acceleration in the expected time. It is also found that  the order of the gate sets $\left\{ L^{2i} \right\}$ is important  to obtain the shortest mean run time of the SWP-DQC algorithm.

\section*{Acknowledgement}
\label{s8}
This work was supported by the National Basic Research Program of China under Grant Nos.~2017YFA0303700 and ~2015CB921001, National Natural Science Foundation of China under Grant Nos.~61726801, ~11474168 and ~11474181, and in part by the Beijing Advanced
Innovation Center for Future Chip (ICFC).

\bibliography{bibfile}

\end{document}